\documentclass[10pt,conference,compsocconf,letterpaper]{IEEEtran}
\usepackage{subfigure,times,hyperref}
\usepackage{amsmath}
\usepackage{amsthm}
\usepackage{amssymb}
\newtheorem{theorem}{Theorem}
\newtheorem{lemma}{Lemma}
\newcommand{\presec}{\vspace{-0.00in}}
\newcommand{\postsec}{\vspace{0.00in}}

\usepackage[ruled,vlined,linesnumbered]{algorithm2e}

\ifCLASSINFOpdf
   \usepackage[pdftex]{graphicx}
   \graphicspath{{../pdf/}{../jpeg/}{./image/}{../eps/}}
   \DeclareGraphicsExtensions{.pdf,.jpeg,.png,.jpg,.eps}
 \else
   \usepackage[dvips]{graphicx}
   \graphicspath{{../eps/}}
  \DeclareGraphicsExtensions{.eps}
\fi

\begin{document}

\title{Locally Self-Adjusting Skip Graphs}

\author{\IEEEauthorblockN{Sikder Huq, Sukumar Ghosh}
\IEEEauthorblockA{Department of Computer Science,
The University of Iowa}}

\maketitle
\begin{abstract}

We present a distributed self-adjusting algorithm for skip graphs that minimizes the average routing costs between arbitrary communication pairs by performing topological adaptation to the communication pattern. Our algorithm is fully decentralized, conforms to the $\mathcal{CONGEST}$ model (i.e. uses $O(\log n)$ bit messages), and requires $O(\log n)$ bits of memory for each node, where $n$ is the total number of nodes. Upon each communication request, our algorithm first establishes communication by using the standard skip graph routing, and then locally and partially reconstructs the skip graph topology to perform topological adaptation. We propose a computational model for such algorithms, as well as a yardstick (working set property) to evaluate them. Our working set property can also be used to evaluate self-adjusting algorithms for other graph classes where multiple tree-like subgraphs overlap (e.g. hypercube networks). We derive a lower bound of the amortized routing cost for any algorithm that follows our model and serves an unknown sequence of communication requests. We show that the routing cost of our algorithm is at most a constant factor more than the amortized routing cost of any algorithm conforming to our computational model. We also show that the expected transformation cost for our algorithm is at most a logarithmic factor more than the amortized routing cost of any algorithm conforming to our computational model.

\end{abstract}


\presec
\section{Introduction} \label{sec: introduction}
\postsec
Many peer-to-peer communication topologies are designed to reduce the worst-case time per operation and do not take advantage of the skew in communication patterns.
Given that most real-world communication patterns are skewed, self-adjustment is an attractive tool that can significantly reduce the average communication cost for a sequence of communications.
For an unknown sequence of communications, self-adjusting algorithms minimize the average communication cost by performing topological adaptation to the communication pattern.
%

In the 1980s, Sleator and Tarjan published their seminal work on \emph{Splay Tree} \cite{SleatorTarjan}, which has been the inspiration of subsequent studies in self-adjusting algorithms, for example \emph{Tango} BSTs \cite{Tango}, \emph{multi-splay trees} \cite{MST}, \emph{CBTree} \cite{Afek}, dynamic skip list \cite{Bose} etc.
All these data structures are designed for centralized lookup operations and rely on a tree or tree-like structure.
Avin et al. proposed \emph{SplayNets} \cite{splayNet}, which initiated the study of self-adjustment for distributed data structures and networks, where each communication involves a source-destination pair.
However, splayNet relies on a single BST structure, and we are not aware of any study of self-adjustment for more complex data structures that relies on the interaction of multiple overlapping tree-like structures (e.g. skip graphs, hypercubic networks).
This motivates our current work on self-adjustment for skip graph topologies.

A skip graph \cite{SG} $G=(V,E)$ is a distributed data structure and a well-known peer-to-peer communication topology that guarantees $O(\log n)$ worst-case communication time between arbitrary pairs of nodes, where $n=|V|$.
The major advantage of skip graphs over BSTs is that skip graphs are highly resilient and capable of tolerating a large fraction of node failures.
In order to achieve such resilience, skip graphs rely on interactions among $n$ overlapping skip list structures.
In general, topological rearrangement of nodes of one skip list affects the structure of multiple other skip lists.
Moreover, since the access pattern in unknown, an adversarial access sequence may incur the worst case communication cost for each of the communication requests.
Thus it is important to keep all the skip lists balanced so that the worst case communication cost for any pair of nodes remains logarithmic.
Now, self-adjusting algorithms generally attempt to move frequently communicating nodes closer to each other.
However, for skip graphs, such an attempt may result in an imbalance situation and drive other uninvolved nodes away from each other, which makes it challenging to design a self-adjusting algorithm for skip graphs.

We present a self-adjusting algorithm \emph{Dynamic Skip Graphs} (\textsf{DSG}) for skip graphs with no {\em a priori} knowledge of the future communication pattern, and analyze the performance of our algorithm.
Upon each communication request, \textsf{DSG} first establishes communication using the the standard skip graph routing in the existing topology, and then \emph{locally} and \emph{partially} transforms the topology to connect communicating nodes via a direct link.
Our algorithm \textsf{DSG} is designed for the $\mathcal{CONGEST}$ model (i.e. allowed message size per link per round is up to $O(\log n)$ bits), and requires $O(\log n)$ bits of memory for each node.
We show that, for an unknown communication sequence, the routing cost for \textsf{DSG} is at most a constant factor more than the optimal amortized routing cost, and the expected transformation cost is at most a logarithmic factor more than the amortized cost of any algorithm conforming to our computational model.


\subsection{Our Contributions}

\begin{enumerate}

\item We propose a computational model for self-adjusting skip graphs. Upon each communication request, our model requires any algorithm to transform the topology such that communicating nodes (the source-destination pair) get connected via a direct link. Our model also limits the memory of each node to $O(\log n)$ bits and conforms to the $\mathcal{CONGEST}$ model.

\item We propose a working set property to evaluate self-adjusting algorithms for skip graphs or similar distributed data structures.

\item We propose a self-adjusting algorithm \emph{Dynamic Skip Graphs} (\textsf{DSG}) conforming to our model and analyze its performance.

\item  Our algorithm uses a distributed and randomized approximate median finding algorithm (\textsf{AMF}) designed for skip graphs. We show that \textsf{AMF} finds an approximate median in expected $O (\log n)$ rounds.

\end{enumerate}

\subsection{Paper organization}

Section \ref{sec:model} presents our self-adjusting model for skip graphs and definitions relevant to this paper. Section \ref{sec:related_work} presents an overview of related work. Section \ref{sec:dsg} and \ref{sec:amf} present algorithm \textsf{DSG} and \textsf{AMF}, respectively. We provide a formal analysis to evaluate our algorithms in Section \ref{sec:analysis}, and conclude this work in Section \ref{sec:conclusion}. Details on some of the steps of our algorithms are presented in the Appendix.

\presec
\section{Related Work} \label{sec:related_work}
\postsec

Interest in self-adjusting data structures grew out of Sleator and Tarjan's
seminal work on {\em splay tree} \cite{SleatorTarjan} that emphasized the
importance of amortized cost and proposed a restructuring heuristics to attain
the amortized time bound of $O(\log n)$ per operation. 
Prior to that, Allen and Munro \cite{AllenMunro}, and Bitner \cite{Bitner} proposed 
two restructuring heuristics for search trees, but none were efficient in the amortized sense.
In \cite{Bagchi}, Bagchi et al. presented an algorithm for efficient access
in {\em biased skip lists} where non-uniform access patterns are biased
according to their weights, and the weights are known. 
Bose et al. \cite{Bose} investigated the efficiency of access in skip lists when the 
access pattern is unknown, and developed a deterministic self-adjusting skip list whose
running time matches the working set bound, thereby achieving dynamic optimality. 
Afek et al. \cite{Afek} presented a version of self-adjusting search trees, called CBTree, 
that promote a high degree of concurrency by reducing the frequency of `` tree rotation.'' 
Avin et al. \cite{splayNet} presented {\em SplayNet}, a generalization of Splay tree, where, 
unlike the splay trees, communication is allowed between any pair of nodes in the tree.
In \cite{OBST}, Avin, et al. extended the concept of splay trees to P2P overlay networks of multiple binary search trees (OBST).
However, their work addresses a routing variant of the classical splay trees that focuses on the lookup operation only. 
In \cite{AvinGrid}, Avin et al. presented a greedy policy for self-adjusting grid networks that locally minimizes an objective function by switching positions between neighboring nodes. 
$\textsf{SKIP}^+$ \cite{skip_plus} presented a self-stabilization (not self-adjusting) algorithm for 
skip graphs. 
\cite{DSG_BA} presents some of the early ideas of our work. 
\presec
\section{The Model and Definition} \label{sec:model} 
\postsec

%

We begin with a quick introduction of Skip Graphs \cite{SG}.
A Skip Graph consists of nodes positioned in the ascending order of their identifiers (often called keys) in multiple levels.
Level 0 consists of a doubly linked list containing all the nodes.
The linked list at level 0 is split into 2 distinct doubly linked lists at level 1.
Similarly, each of the linked lists at level 1 is split into 2 distinct linked lists at level 2, and this continues recursively for upper levels until all nodes become singleton.
In other words, every linked list with at least 2 nodes at any level $i$ is split into 2 distinct linked lists at level $i+1$.
The number of levels in a skip graph is called the \emph{height} of the skip graph.
When a linked list splits into 2 linked list at the next upper level, we denote the split linked lists as \emph{0-sublist} (or 0-subgraph) and \emph{1-sublist} (or 1-subgraph).
Note that we refer the base (lowest) level as level 0.
We denote the height of a skip graph as $H$.

\begin{figure}[!t]
\def \subfigcapskip{0pt}
\centering
    \subfigure[A skip graph with 6 nodes and 3 levels.]
        {\includegraphics[width=0.8\columnwidth]{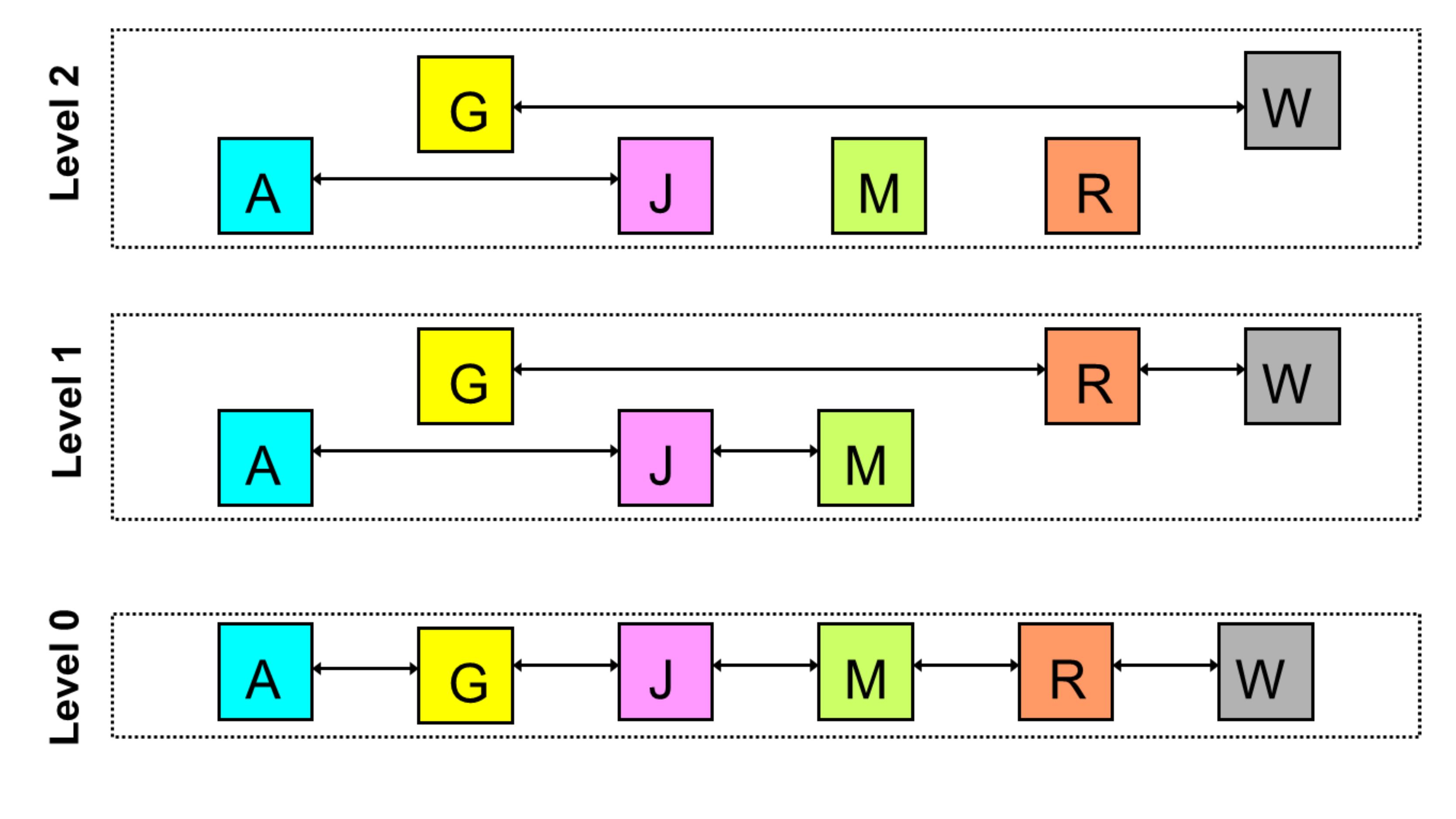}}

\centering
    \subfigure[The skip graph in (a) represented by a binary tree.]
        {\includegraphics[width=0.8\columnwidth]{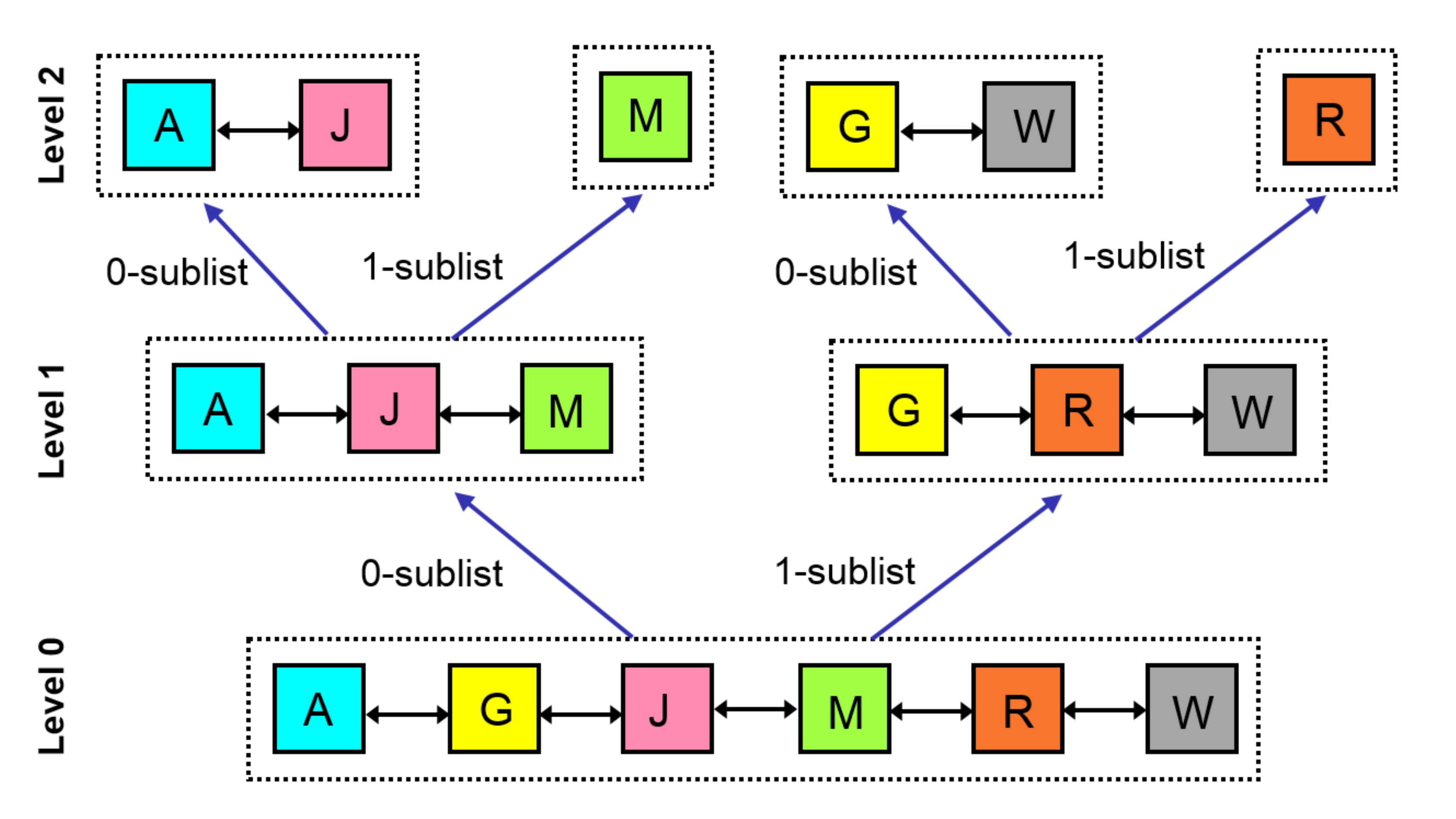}}

    \caption{The lowest 3 levels of a skip graph shown in (a) is mapped to an equivalent tree structure in (b).}
    \label{fig:sg_to_tree}

\end{figure}

For simpler representation, we map a skip graph into a binary tree of linked lists.
To this end, the linked list at level 0 is represented by the {\em root} node of the tree,
and the 0-sublist and the 1-sublist at level 1 are represented by the {\em left child} and
{\em right child} of the root, respectively.
Similarly, every linked list of the skip graph is represented by a node in the binary tree,
mimicking the parent child relationship of the skip graph in that of its equivalent binary tree.
Figure \ref{fig:sg_to_tree}(a) shows a skip graph with 3 levels, and figure \ref{fig:sg_to_tree}(b)
shows the corresponding binary tree representation.

Each node $x$ in a skip graph has a {\em membership vector} $m(x)$ of size $H-1$.
The $i^{th}$ bit of $m(x)$ represents the sublist (0 or 1) that contains node $x$ at level $i$.
For example, the membership vector of node $M$ in figure \ref{fig:sg_to_tree}(b) is 01, as $M$ belongs
to the 0-sublist at level 1, and 1-sublist at level 2.

Every node in the binary tree is the root of a subtree that represents a sub(skip)graph (or sub graph)
of the skip graph.
We refer the subgraph rooted by a 0-sublist and 1-sublist as 0-subgraph and 1-subgraph, respectively.
Since the construction is recursive, we can also designate a subgraph as $b$-subgraph, where $b$ is a
bit string containing the common prefix bits of the membership vectors for all nodes in the subgraph.
For example, the subgraph containing nodes $G$ and $W$ in figure \ref{fig:sg_to_tree}(b) is
designated by $10$-subgraph.

Let $V=\{1,...,n\}$ be a set of nodes (or peers).
Let $\mathcal{S}$ be the family of all Skip Graphs of $n$ nodes, where each topology $G(V,E) \in
\mathcal{S}$ is a skip graph with $O(\log n)$ levels.
For any skip graph $S \in \mathcal{S}$, $L_i$ denotes the set of all linked lists at level $i$ of $S$.
We define the following balance property that must hold for the family of skip graphs $\mathcal{S}$:

\noindent
\textbf{Definition (\textit{a-balance Property}).} A Skip Graph satisfies the \textit{$a$-balance} property
if there exists a positive integer $a$, such that among any $a+1$ consecutive nodes in any linked
list $l \in L_i$, at most $a$ nodes can be in a single linked list in $L_{i+1}$. The $a$-balance property
ensures that the length of the search path between any pair of nodes is at most $a \cdot \log n$.

\noindent
\textbf{Self-Adjusting model for Skip Graphs.} We consider a synchronous computation model, where
communications occure in \emph{rounds}. A node can send and receive at most 1 message through a link
in a round (i.e. $\mathcal{CONGEST}$ model). Our model limits the memory of each node to $O (\log n)$
bits. Given a skip graph $S \in \mathcal{S}$, and a pair of communicating nodes $(u,v) \in V \times V$,
a self-adjusting algorithm performs the followings:

\begin{enumerate}
\item Establishes communication between nodes $u$ and $v$ in $S$.
\item Transforms the skip graph $S$ to another skip graph $S' \in \mathcal{S}$, such that nodes
$u$ and $v$ move to a linked list of size two at any level in $S'$. This implies that a direct link
needs to be established between nodes $u$ and $v$.
\end{enumerate}

Note that the height of $S'$ must be $O(\log n)$ since $S' \in \mathcal{S}$. Let
$\sigma = (\sigma_1, \sigma_2, ... , \sigma_{m})$ be an unknown access sequence consisting of $m$
sequential communication requests, $\sigma_t = (u,v) \in V \times V$ denotes a routing request from
source $u$ to destination $v$. Given a skip graph $S$, we define the distance $d_S(\sigma_t)$ as the
number of intermediate nodes in the communication path from the source to destination associated with
request $\sigma_t$, where the communication path is obtained by standard skip graph routing
algorithm \cite{SG}. An overview of the standard skip graph routing is presented in
Appendix \ref{sec:appen_sg_routing}.

Given a request $\sigma_{t}$ at time $t$, let an algorithm $\mathcal{A}$ transforms the current
skip graph $S_{t}\in \mathcal{S}$ to $S_{t+1} \in \mathcal{S}$. We define the cost for network
transformation as the number of rounds needed to transform the topology. We denote this transformation
cost at time $t$ as $\rho( \mathcal{A}, S_{t}, \sigma_{t})$. Similar to a prior work \cite{splayNet},
we define the cost of serving request $\sigma_{t}$ as the distance from source to destination plus the
cost of transformation performed by $\mathcal{A}$ plus one, \textit{i.e.},
$d_{S_t}(\sigma_t)+\rho( \mathcal{A}, S_{t}, \sigma_{t})+1$.

\noindent
\textbf{Definition (\textit{Average and Amortized Cost}).} We used definitions similar to
\cite{splayNet} here. Given an initial skip graph $S_0$, the average cost for algorithm
$\mathcal{A}$ to serve a sequence of communication requests
$\sigma = (\sigma_1, \sigma_2, \cdots, \sigma_m)$ is:

\begin{equation}
Cost( \mathcal{A}, S_{t}, \sigma_{t}) = \frac{1}{m} \sum_{i=1}^{m} (d_{S_t}(\sigma_t)+\rho( \mathcal{A}, S_{t}, \sigma_{t})+1)
\end{equation}

The amortized cost of $\mathcal{A}$ is defined as the worst possible cost to serve a
communication sequence $\sigma$, \textit{i.e.}
\\ $max_{S_0,\sigma}$ Cost$(\mathcal{A}, S_0, \sigma)$.

\noindent
\textbf{Definition (\textit{Sub Skip Graph}).} A sub skip graph (often called subgraph in this paper) is
a skip graph that is a part of another skip graph. In other words, given a skip graph $S (V, E)$, a sub
skip graph $S^\prime (V^\prime, E^\prime)$ is a skip graph in $S$ such that $V^\prime \subseteq V$ and
$E^\prime$ is the set of links from $S$ induced by the nodes in $V^\prime$. We call a sub skip graph is
at level $d$ when all nodes in the sub skip graph share a common membership vector prefix of size $d$.
We call the lowest level of a sub skip graph as the \textit{base} level, and the linked list that
contains the nodes of a sub skip graph as the base linked list for that sub skip graph. Observe
that there exists a sub skip graph for all linked lists in any skip graph.

\noindent
\textbf{Definition (\textit{Working Set Number}).} We denote the \textit{working set number} for request
$\sigma_i$ as $T_i(\sigma_i)$. Let $u_i$ be the source and $v_i$ be the destination specified by
communication request $\sigma_i$. To define $T_i(\sigma_i)$, we construct a communication graph $G$ with
the nodes that communicated (either as source or destination) during the time period starting from the
last time $u_i$ and $v_i$ communicated, and ending at time $i$. We draw an edge between any two nodes
in $G$ if they communicated in this time duration. Now, we define the \textit{working set number} for
request $\sigma_i$, $T_i(\sigma_i)$, as the number of distinct nodes in $G$ that have a path from
either $u_i$ or $v_i$. In case $u_i$ and $v_i$ are communicating for the first time, $T_i(\sigma_i) = n$
by default, where $n$ is the number of nodes in the skip graph.

As an example, for the latest communication request $(u,v)$ shown in figure \ref{fig:wset}(a),
the corresponding communication graph $G$ is shown in figure \ref{fig:wset}(b). The number of distinct
nodes in $G$ that have a path from either $u$ or $v$ is 5; therefore the working set number for
the communication request is 5.

\noindent
\textbf{Definition (\textit{Working Set Property}).} For a skip graph $S$ at time $i$, the
\emph{working set property} for any node pair $(x,y) \in S$ holds iff $d_S(x,y) \leq \log T_i(x,y)$.
Note that, the $\log$ in the definition is to address the tree-like structure of the skip graph topology.

\noindent
\textbf{Definition (\textit{Working Set Bound}).} We define the \emph{working set bound} as
$WS(\sigma) = \sum_{i=1}^{m}$ log$( T_i(\sigma_i))$.

\begin{figure}[!t]
\def \subfigcapskip{0pt}
\centering


    \subfigure[An access pattern showing a repeating communications between $u$ and $v$.]
    {\includegraphics[width=0.8\columnwidth]{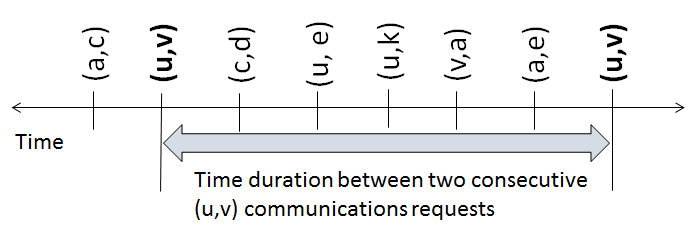}}

\centering
    \subfigure[Communication graph $G$ for the time duration shown in (a).]
    {\includegraphics[width=0.4\textwidth]{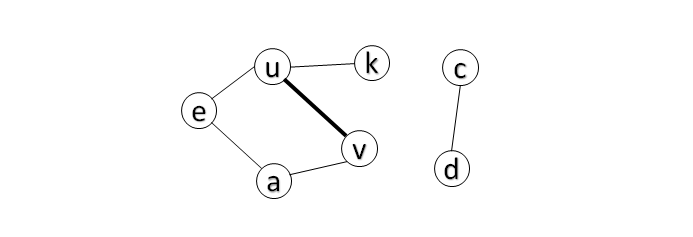}}

    \caption{For the access pattern shown in (a), the working set number for the last communication between $u$ and $v$ is 5, as the the number of distinct nodes in $G$ that has a path from either $u$ or $v$ is 5 (e,a,k,u and v).}

    \label{fig:wset}
\end{figure}

\begin{theorem}
\label{WSTheorem3}
For an unknown communication sequence $\sigma = \sigma_1, \sigma_2, \dots, \sigma_m$, the amortized
routing cost for any self-adjusting algorithm conforming to our model is at least $WS(\sigma)$ rounds.
\end{theorem}


\begin{proof}
Let us assume that the theorem does not hold. Then there must be at least one request
$\sigma_i = (u_i, v_i)$ such that $d_{S_i} (u_i, v_i) < T_i(u_i, v_i)$. However, due to the construction
of the skip graph, this results in existance of a node $w_i$ such that $d_{S_i} (u_i, w_i) > T_i(u_i, w_i)$.
An example to illustrate this idea is presented right after this proof.

Since the communication sequence is unknown, it is possible that $(u_i, w_i)$ is chosen as
request $\sigma_i$ instead of $(u_i, v_i)$. Since this argument is applicable for any $\sigma_i \in \sigma$,
the theorem holds.
\end{proof}

\noindent
\textbf{Example of Working Set Bound.}
Consider the communication graph in Figure \ref{fig:ws_proof}. Let $l$ be a linked list of size
$2k$ in a skip graph. Let nodes $U, V, A, A_1, A_2, \cdots, A_{n-2}$ belong to the linked list $l$.
Let the linked list $l$ split into 2 subgraphs (i.e. sublists) each with size $k$ at the next (upper) level.
Suppose node $A$ moves to the 0-subgraph. Now, we need to choose other $k-1$ nodes to accompany node $A$
in the 0-subgraph.

Let we move nodes $U$ and $V$ to the 0-subgraph. Then there exists a node $A_i, 1 \leq i \leq k-2$ that
moves to the 1-subgraph. Clearly this violates the working set property for the pair $(A, A_i)$. However,
if we move nodes $U$ and $V$ to the 1-subgraph, we violate the working set property for the pair $(A,U)$.
Thus, $U$ must move to the 0-subgraph and $V$ must move to the 1-subgraph. Note that, at time $t^\prime + k$,
the working set number for pair $(U,V)$, $T_{t^\prime + k} (U,V) = k + 1$; and the routing distance for pair
$(U,V)$ is $\left \lceil{T_{t^\prime + k} (U,V)}\right \rceil =
\left \lceil{\log(k+1)}\right \rceil = \log _2 (2k)$.

\begin{figure}[!t]
\centering
\includegraphics[width=0.6\columnwidth]{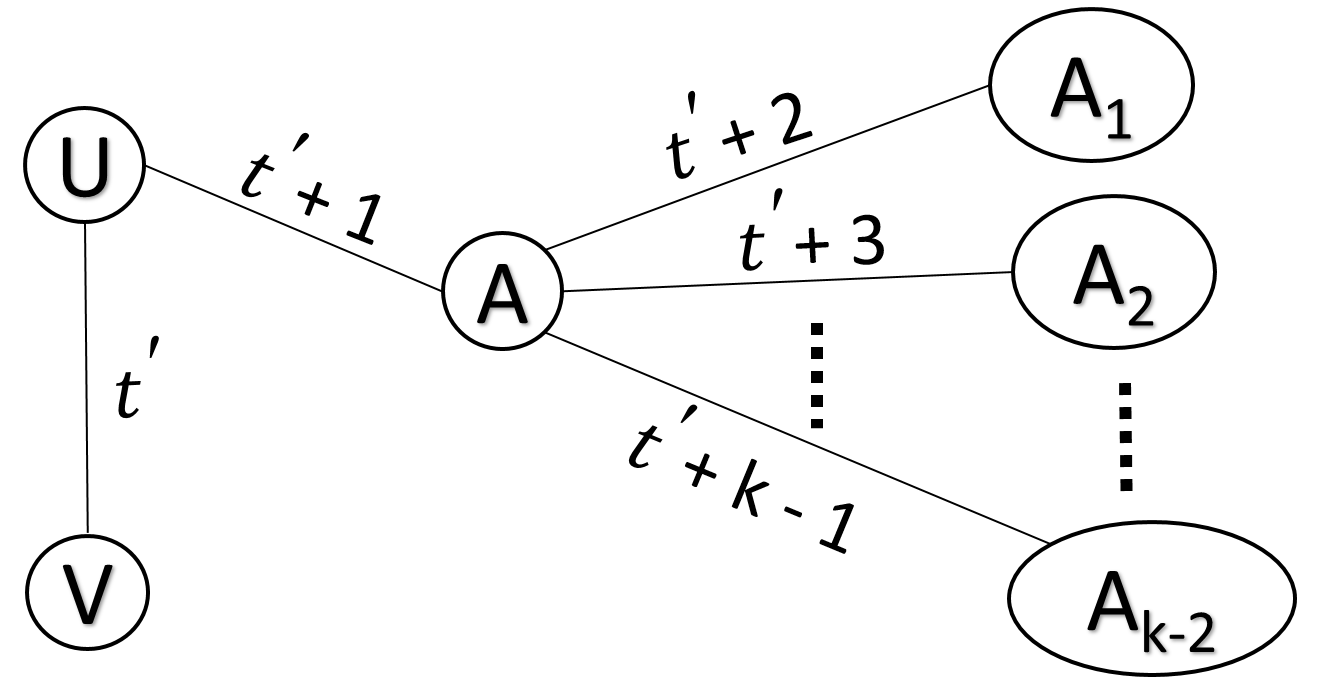}

\caption{A communication graph. Each edge is labeled with the timestamp of the most recent communication
between the end nodes.}
\label{fig:ws_proof}

\end{figure}

\presec
\vspace{-0.0in}
\section{Dynamic Skip Graphs (DSG)} \label{sec:dsg}
\postsec

\subsection{Overview} \label{sec:overview}
\postsec

Upon a communication request, our algorithm \textsf{DSG} first establishes the communication using the standard skip graph routing, then performs atomic topological transformation conforming to the self-adjusting  model. The key idea behind \textsf{DSG} is that frequently communicating nodes form groups at different levels and a node's attachment to a group is determined by a timestamp. Each node has a \emph{group-id} and a \emph{timestamp} associated with each level. A node is a member of a \emph{group} at each level, and the group-id is an identifier that represents a group. All nodes belong to the same group at a level hold the common group-id. The timestamp associated with a level is used by \textsf{DSG} to identify how attached a node is with its group at that level.

When two nodes from two different groups communicate, \textsf{DSG} merges the communicating groups to a single group. However, when a group grows too big to be accommodated in a single linked list at the corresponding level (due to the structural constraint of the skip graph), \textsf{DSG} splits the group into two smaller groups. There are three challenges here. First, since the goal of \textsf{DSG} is to ensure that the working set property always holds for any node pair in any group, distance between any two nodes from any non-communicating group should not increase. In other worlds, routing distances among nodes of any non-communicating groups should not be affected due to a transformation. Second, the working set property must hold for any node pair in the merged and split groups after a transformation. Third, a transformation requires a partial reconstruction of the skip graph structure. According to our self-adjusting model, the height of the skip graph must remain $O(\log n)$ after any reconstruction.

Transformation starts from the highest level at which communicating nodes share a common linked list. For example, the highest level with a common linked list for nodes $A$ and $M$ in the skip graph in Figure \ref{fig:sg_to_tree} is level 1 and the common linked list is the linked list that contains only nodes $A$, $J$ and $M$. Starting from the highest level with common linked list, transformation continues recursively and parallelly in the upper levels until all the involved nodes become singleton (i.e. move to a linked list of size 1). For each of the newly created linked lists of size $>1$, transformation takes place as follows:

\begin{enumerate}

\item[--] Each node of the linked list computes a priority using certain priority rules. Each priority is a function of node's group-id and timestamp for the corresponding level. The priorities are computed in a way such that all groups have a distinct range of priorities. The communicating nodes have the highest priority, each node of the merged group has a positive priority, and all other nodes have a negative priority.

\item[--] All nodes of the linked list compute an approximate median priority using the algorithm \textsf{AMF}. In general, at the next upper level after transformation, any node with a priority higher or equal to the approximate median priority moves to the 0-subgraph, and any node with a priority lower than the approximate median priority moves to the 1-subgraph. \textsf{DSG} uses priorities to ensure that nodes from the same group remain together after a transformation. However, a transformation technique based on comparing priority with approximate median priority may split a non-communicating group. Such cases are handled carefully by \text{DSG}. Note that, the approximate median priority is used to ensure that the sizes of the 0-subgraph and 1-subgraph are roughly the same after a transformation, keeping the height of the skip graph always $O(\log n)$.

\end{enumerate}

Since communicating nodes always move to the 0-subgraph, after transformation in all levels, communicating nodes are guaranteed to move to a linked list of size 2. Each node involved in the transformation reassigns its group-ids and timestamps such that \textsf{DSG} can work consistently for future communication requests.

\presec
\subsection{Setup and notations} \label{sec:setup}
\postsec

Let $H_t$ be the height of the skip graph at time $t$. \textsf{DSG} requires every node to hold $H_t$ bits to store its membership vector. In addition, each node stores a {\em timestamp} and a {\em group-id} for each of the levels. For node $i$ and level $j$, we use the notations $V^i_j$, $T^i_j$ and $G^i_j$ to denote the stored membership vector bit, timestamp, and group-id respectively. Initially, all timestamps are set to zero and all group-ids are set to the corresponding node's identifier.

\begin{figure}[!t]

    \def \subfigcapskip{0pt}

    \subfigure[A communication graph $G$]
        {\includegraphics[width=0.2\textwidth]{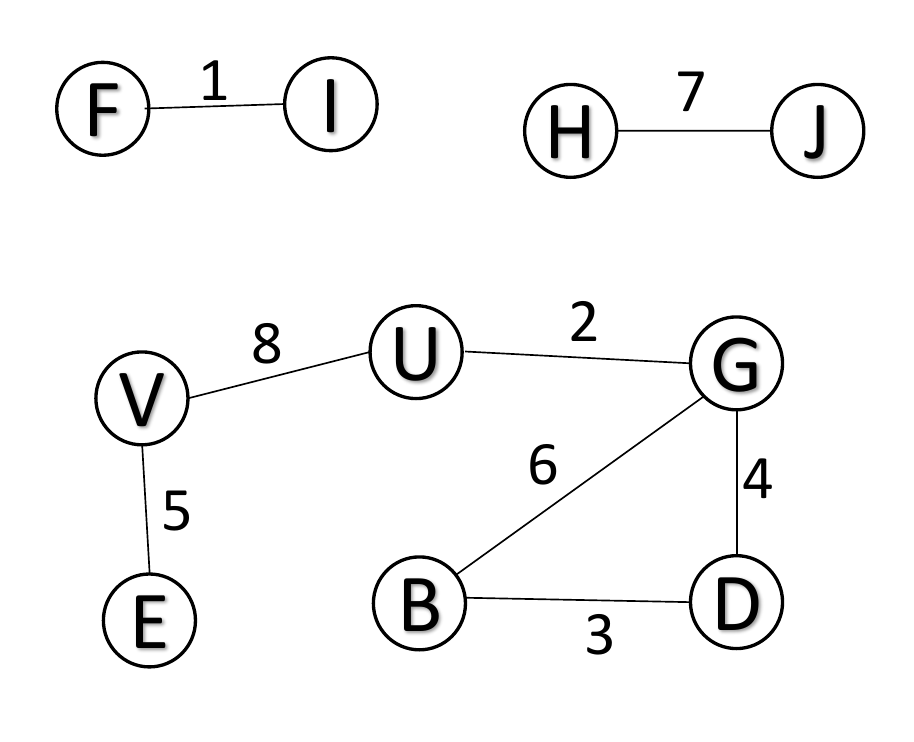}}

    \vspace{0.2in}

    \subfigure[Skip graph at time 8, $S_8$.]
        {\includegraphics[width=0.5\textwidth]{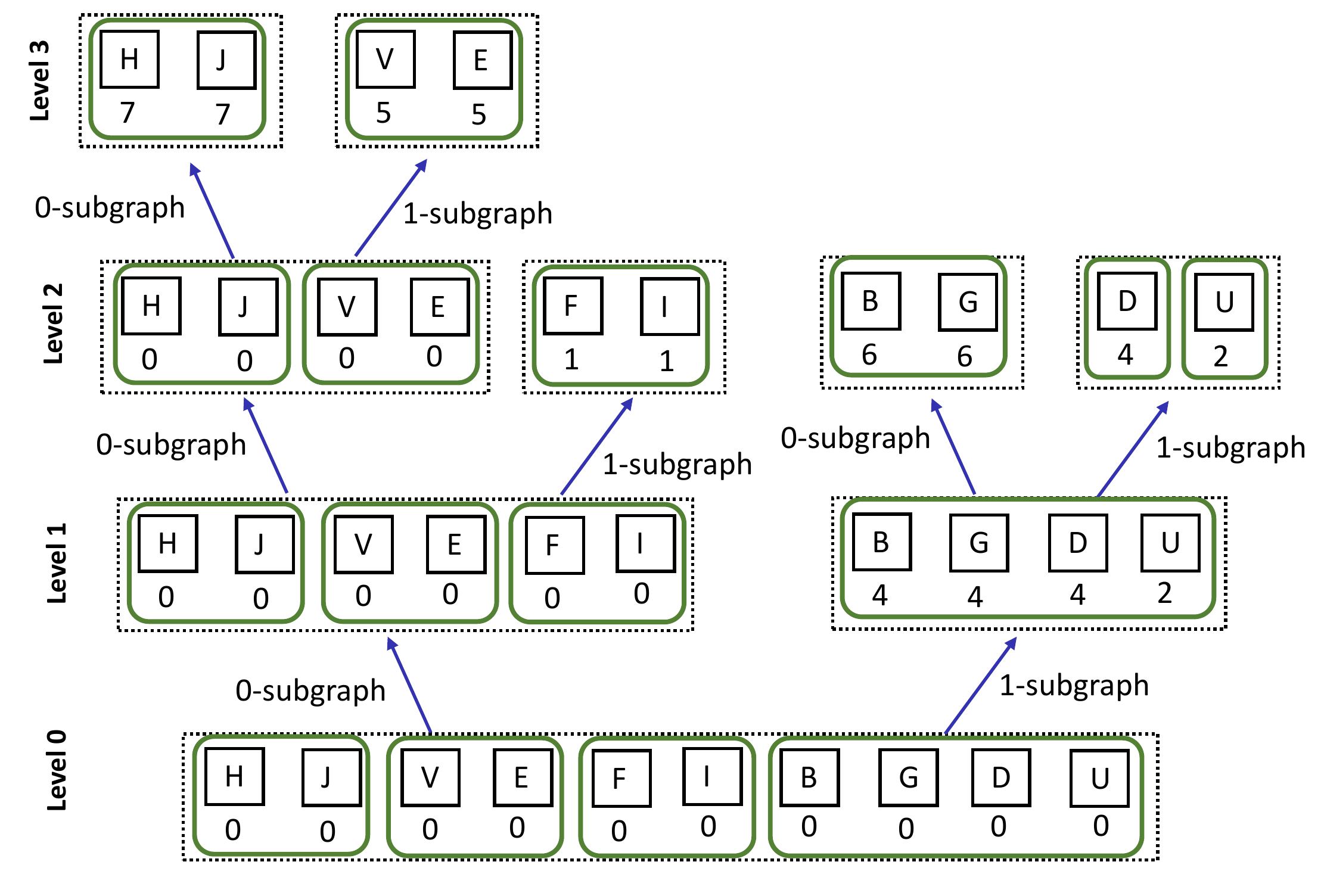}}

    \vspace{0.2in}

    \subfigure[Skip graph at time 9, $S_9$.]
        {\includegraphics[width=0.5\textwidth]{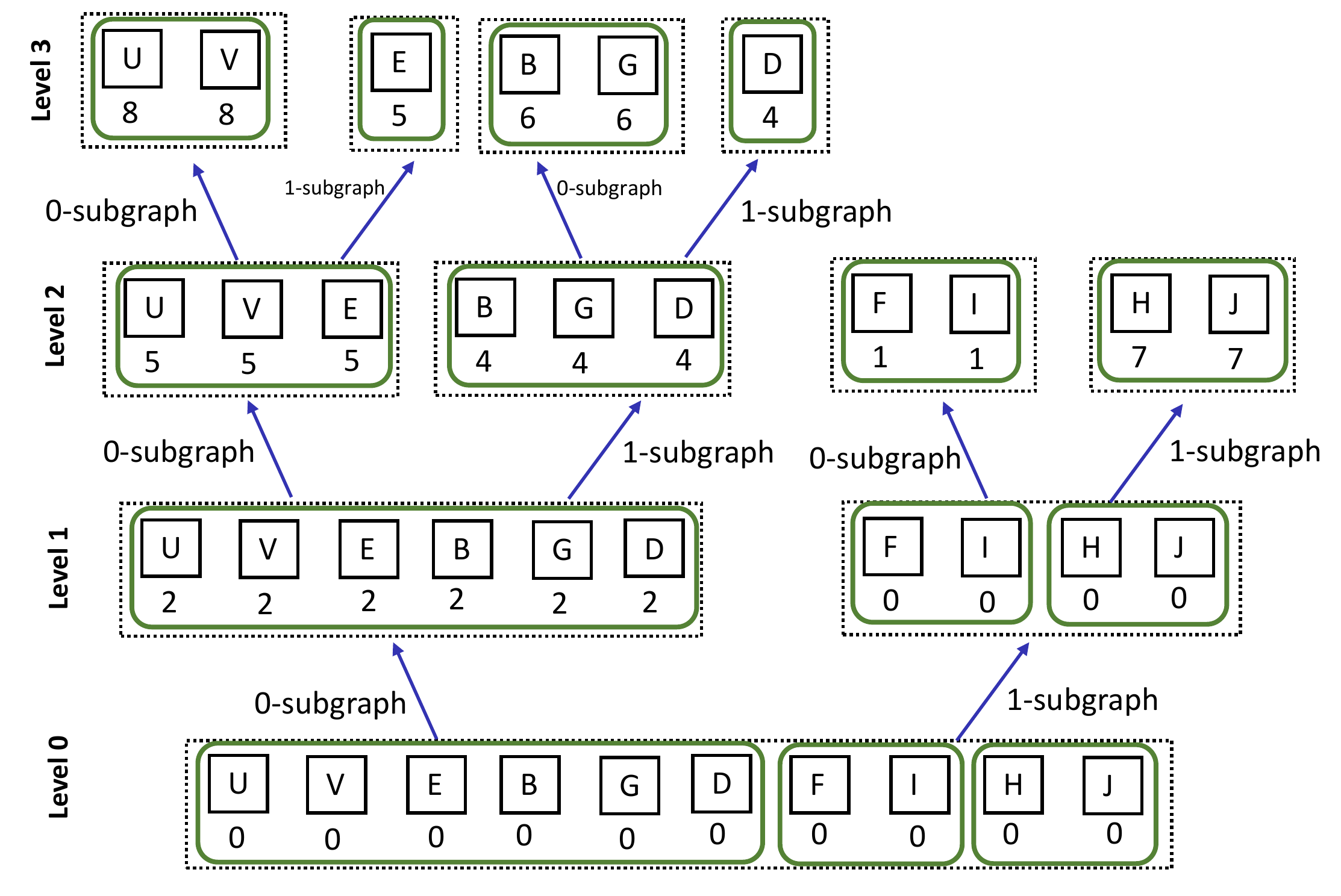}}

    \vspace{0.2in}

    \caption{For the communication graph $G$ shown in (a), possible skip graph representations at time 8 and 9 ($S_8$ and $S_9$), obtained by $\textsf{DSG}$, are shown in (b) and (c), respectively. The rounded rectangles show the groups of nodes at different levels, and the number below each node is the timestamp for the node at the corresponding level. For example, in $S_9$ (figure c), the group of node $B$ at level 2 has 3 nodes ($B$, $G$, and $D$), and the timestamp of node $B$ for level 2 (i.e. $T^B_2)$ is 4.}
    \label{fig:example}
\end{figure}

Figure \ref{fig:example}(a) shows a communication graph $G$, where each edge is labeled with the time associated with the most recent communication. Observe that nodes $U$ and $V$ communicate at time 8 as shown in $G$. For the communication graph $G$, Figure \ref{fig:example}(b) shows a binary tree representation of a valid skip graph ($S_8 \in \mathcal{S}$) at time 8 obtained by \textsf{DSG}. Figure \ref{fig:example}(c) shows a valid skip graph ($S_9 \in \mathcal{S}$) representation, where $S_9$ is transformed from $S_8$ using  \textsf{DSG}, as a result of the communication between nodes $U$ and $V$ at time 8. The numbers below each node in Figure \ref{fig:example}(b) and (c) give the timastamp of the node at the corresponding level. We use this transformation from $S_8$ to $S_9$ as an example for the description of our algorithm for the remaining paper.

For a communication request from node $u$ to node $v$, let $\alpha$ be the highest common level of the current skip graph with a linked list containing both nodes $u$ and $v$. Let $l_\alpha$ denote that linked list, then $u,v \in l_\alpha$. Let $t$ be the time when the request is originated, and $S_t$ be the skip graph at time $t$. In Figure \ref{fig:example}(b), for nodes $U$ and $V$, $\alpha = 0$, $l_\alpha$ is the linked list represented by the root of the binary tree, and $t = 8$.

We explain different parts of our algorithm in detail in following subsections.

\presec
\subsection{Transformation from $S_t$ to $S_{t+1}$} \label{sec:transformation}
\postsec

Upon a routing from node $u$ to node $v$, node $v$ records the highest common level number $\alpha$ and shares $\alpha$ with node $u$. Then both nodes $u$ and $v$ broadcast a transformation notification to all nodes in $l_\alpha$. The notification message includes all $H_t$ timestamps, group-ids and membership vectors of nodes $u$ and $v$. All nodes  $x \in l_\alpha$ compute a priority $P(x)$ by using following priority-rules:

\begin{enumerate}

\item[P1:] (Rule for communicating nodes). Nodes $u$ and $v$ set $\infty$ as their priority. In other words, set $P(u) = P(v) = \infty$.

\item[P2:] (Rule for nodes in the same group of either communicating nodes at level $\alpha$). All nodes $x \in l_\alpha, x \neq u, x \neq v, G^x_\alpha = G^u_\alpha$ assign their priority $P(x) = \min (T^x_c, T^u_c)$ where $c$ is the highest level in $S_t$ such that $G^x_c = G^u_c$. Similarly, all nodes $x \in l_\alpha, x \neq u, x \neq v, G^x_\alpha = G^v_\alpha$ assign $P(x) = \min (T^x_c, T^v_c)$ where $c$ is the highest level in $S_t$ such that $G^x_c = G^v_c$.

\item[P3:] (Rule for other nodes). Each node $x \in l_\alpha, x \neq u, x \neq v, G^x_\alpha \neq G^u_\alpha, G^x_\alpha \neq G^v_\alpha$ (i.e. neither in $u$'s nor in $v$'s group at level $\alpha$) set $P(x) = - (G^x_\alpha \cdot t)   + T^x_{\alpha + 1} $.

\end{enumerate}

We require that group identifiers are non-negative integers (possibly an ip address of a node). Observe that, all nodes of the communicating group at level $\alpha$ have a positive priority as a timestamp is always positive, and rest of the nodes have a negative priority as \textsf{DSG} ensures that $t > T^x_{\alpha + 1}$. Also, according to P3, priorities assigned to nodes from a non-communicating group range between $ - (G^x_\alpha \cdot t)$ and $-(G^x_\alpha+1) \cdot t$, where $G^x_\alpha$ is the group-id.

In our example in Figure \ref{fig:example}, as nodes U and V communicate at time 8, the $\alpha$ is 0 ($S_8$ in Figure \ref{fig:example}(b)). Let us assume that the nodes' numerical identifiers are determined by their positions in the English alphabet, e.g. identifier for node A is 1, identifier for node B is 2, and so on. According to the priority rule P1, $P(U) = P(V) = \infty$; and according to the priority rule P2, $P(D) = 2$, $P(G) = 2$, $P(B) = 2$ and $P(E) = 5$. Let us assume that the group-ids $G^H_0 = G^J_0 = 10$ (as H is the tenth letter in alphabet) and $G^F_0 = G^I_0 = 6$. Hence, according to the priority rule P3, $P(H) = P(J) =  -(10 \times 7) + 2 = -68$, and $P(F) = P(I) = -(6 \times 7) + 2 = -40$.

At this point, node $u$'s group at level $\alpha$ merges with node $v$'s group at the same level by updating their group-ids. To this end, all nodes $x \in l_\alpha$ with $G^x_\alpha = G^u_\alpha$ or $G^x_\alpha = G^v_\alpha$  set $G^x_\alpha = u$. Note that, by $u$, we mean the identifier of node $u$.

Transformation begins at level $\alpha + 1$ and recursively continues at upper levels. Only the nodes in linked list $l_\alpha$ take part in the transformation. For the remaining section, we write $d$ to refer the current level of transformation and $l_{d-1}$ to refer the linked list that is involved in the transformation. Initially $d=\alpha + 1$ and $l_{d-1}=l_\alpha$.

All nodes in $l_{d-1}$ find an approximate median priority to decide whether to move to the 0-subgraph or to the 1-subgraph (i.e. determine new membership vector bit $V^x_{d}$ in the new skip graph $S_{t+1}$). We propose a distributed approximate median finding algorithm (\textsf{AMF}) for skip graphs to find the approximate median in expected $O(\log_a n)$ rounds, where $a$ is a constant. The algorithm \textsf{AMF} is described in section \ref{sec:amf}. For now let us consider \textsf{AMF} as a black box that finds an approximate median priority and broadcasts it to all nodes in linked list $l_{d-1}$.

However, to utilize \textsf{AMF}, in some cases we require to identify the nodes that moved to the 0-subgraph by receiving a positive approximate median priority. To this end, we introduce a set of boolean variables referred to as \emph{is-dominating-group}s, held by each node of the skip graph. Each node holds $H_t$ is-dominating-group variables, one for each level. Let $D^x_d$ denote the is-dominating-group of node $x$ for level $d$. The goal is to ensure that any node $x$ with $D^x_d = True$ moved to the 0-subgraph at level $d+1$ in past when it received a positive approximate median priority at level $d$ for the last time.

Let the approximate median priority be $M$. One of the following cases must follow:

\textbf{Case 1 ($M$ is positive).} Each node $x$ with $P(x) \geq M$ moves to the 0-subgraph at level $d$ and sets is-dominating-group $D^x_d = True$. Each node $x$ with $P(x) < M$ moves to the 1-subgraph at level $d$ and sets is-dominating-group $D^x_d = False$. According to the priority rules P1 and P2, this case splits the merged group of nodes $u$ and $v$.

\textbf{Case 2 ($M$ is Negative).} When $M$ is negative, there may exist a group $g_s$ such that all nodes $x \in g_s$ finds the following condition true.

\begin{equation} \label{eq:gr_split}
- G^x_{d-1} \cdot t  \geq M \geq - (G^x_{d-1}+1) \cdot t
\end{equation}

If that happens, splitting group by comparing $P(x)$ and $M$ (as we do for case 1) may split the group $g_s$ at level $d$ as its nodes may move to different subgraphs. Given that $M$ is negative, the priority rule P3 confirms that group $g_s$ is a non-communicating group at level $d$ (i.e. contains neither $u$ nor $v$). Clearly the working set property will be violated if the group $g_s$ is split at level $d$ since it will increase the distance between some node pairs in group $g_s$. To fix this issue, nodes in $l_d$ perform a distributed count to compute $|l_d|$ and $|g_s|$. Then nodes decide which subgraph to move to as follows:

\begin{itemize}

\item[--] If $|g_s| > \frac{2}{3} |l_d|$ ($|g_s|$ is too big, thus we need to split $g_s$)
    \begin{itemize}
    \item any node $x \in g_s$ moves to the 1-subgraph if $D^x_d = True$; $x$ moves to the 0-subgraph otherwise.
    \item any node $x \in l_{d-1} \text{ and } x \not \in g_s$ moves to the 0-subgraph.
    \end{itemize}

\item[--] If $|g_s| < \frac{1}{3} |l_d|$ ($|g_s|$ is sufficiently small, thus we move all nodes of $g_s$ either to the 0-subgraph or to the 1-subgraph)
    \begin{itemize}
    \item any node $x$ such that $x \in l_{d-1} \text{ and } x \not \in g_s$ moves to the 0-subgraph if $P(x) \geq M$; $x$ moves to the 1-subgraph otherwise.
    \item Let $L_{low} = \{x \in l_d | P(x) < M\}$ and $L_{high} = \{x \in l_d | P(x) \geq M\}$. Clearly, $|l_d| = L_{low} + L_{high}$. Any node $x \in g_s$ moves to the 0-subgraph if $L_{high} < L_{low}$; $x$ moves to the 1-subgraph otherwise.
    \end{itemize}

\item[--] If $\frac{1}{3} |l_d| \leq |g_s| \leq \frac{2}{3} |l_d|$ (Move all nodes of $g_s$ to the 1-subgraph and rest of the nodes to the 0-subgraph)
    \begin{itemize}
    \item any node $x$ such that $x \in l_{d-1} \text{ and } x \not \in g_s$ moves to the 0-subgraph.
    \item any node $x \in g_s$ moves to the 1-subgraph.
    \end{itemize}

\end{itemize}

Note that $|g_s|$, $L_{low}$ and $L_{high}$ can be computed in $O(\log n)$ rounds by computing distributed sum using a balanced skip list. Algorithm \textsf{AMF} constructs a balanced skip list to compute the median priority. The balanced skip list created by \textsf{AMF} can be reused to compute $|g_s|$, $L_{low}$ and $L_{high}$. To avoid distraction, we present the distributed sum algorithm in Appendix \ref{sec:appen_dist_count}.

As nodes decide which subgraph to move to, two new linked lists are formed at level $d$. To find the left and right neighbors at level $d$, nodes linearly search for neighbors at level $d-1$. Because of the a-balance property, it is guaranteed that a node finds both of its left and right neighbors in at most $a$ rounds (end-nodes have just one neighbor). All nodes $x$ in the new linked lists that do not contain communicating nodes $u$ and $v$ recompute $P(x)$ (for upcoming transformation at level $d$) using the priority rule P4. Computation of $P(x)$ with P4 is similar to P3 except that $\alpha$ is replaced by $d$.

\begin{enumerate}

\item[P4:] (Rule for nodes moved to a linked list that does not contain nodes $u$ and $v$). Each node $x$ moves to a linked list $l_d$ such that $u,v \neq l_{d-1}$ sets $P(x) = - (G^x_d \cdot t)   + T^x_{d+1} $.

\end{enumerate}

Since the communicating nodes $u$ and $v$ always move to the 0-subgraph, a linked list $l_d$ contains nodes $u$ and $v$ only if the following condition is true for all nodes $x \in l_d$.
\begin{equation} \label{eq:uv_check}
V^x_{\alpha+1} = V^x_{\alpha + 2} = ... = V^x_d = 0
\end{equation}

The transformation procedure described above is performed recursively and parallelly by each new linked list until all nodes become members of a singleton list. Clearly this ensures that node $u$ and $v$ get connected directly, and all nodes $x \in l_\alpha$ finds their new and complete membership vectors.

Getting back to our example in Figure \ref{fig:example}, let the approximate median priority $M$ be 2 for $d = 1$. Then nodes U, V, E, B, G and D move to the 0-subgraph and nodes F, I, H and J move to the 1-subgraph at level 1. Now nodes in the 1-subgraph recomputes their priority by using P4. Then both linked lists at level 1 performs the transformation recursively and paralelly, and this continues at upper levels.

In the following two subsections, we discuss how nodes involved with a communication reassigns their group-ids and timestamps.

\presec
\subsection{Assignment of new group-ids} \label{sec:newGID}
\postsec

For each level $d = \alpha, \alpha+1, \cdots , H_{t+1}$, each node $x \in l_{d-1}$ checks if nodes $u$ and $v$ belong to their own linked list at level $d$ by checking the condition in equation \ref{eq:uv_check}. If a node finds the condition in equation  \ref{eq:uv_check} to be
true, the node sets its group-id $G^x_d$ to the identifier of node $u$. If the condition in equation \ref{eq:uv_check} is found to be false by any node, it checks if its group at level $d$ is getting split due to the transformation. A group without communicating nodes $u$ and $v$ can split at level $d$ only if the group is $g_s$ and $|g_s| > \frac{2}{3} |l_d|$.
In case of such a split, the identifier of the left-most node in the split group is used as the new group-id. Observe that it is easy for the left-most node to detect itself since it does not find a left neighbor. The balanced skip list structure created by algorithm \textsf{AMF} is then reused to propagate the new group-id to all the nodes of the group. The left-most node first sends its identifier to the head node of the balanced skip list, and then the identifier is broadcasted to all nodes at the base level of the skip list. The balanced skip list is destroyed once the new group-id is sent to all members of the group.

As nodes $u$ and $v$ move to the same group, for correctness, it is necessary that all nodes in node $u$'s and node $v$'s group at any level $d < \alpha$ in $S_t$ must have the same group-id at level $d$ in $S_{t+1}$. All such nodes update their group-ids for levels below $\alpha$ by using the procedure presented in Appendix \ref{sec:appen_group_id}, only if $G^u_{\alpha-1} \neq G^v_{\alpha-1}$.

\presec
\subsection{Assignment of new timestamps} \label{sec:newTS}
\postsec

Each node $x \in l_\alpha$ updates its timestamps by using the following timestamp-rules (executes in the order given below):
\begin{enumerate}

\item[T1:]  Let $d^\prime$ be the level at which nodes $u$ and $v$ form a linked list of size 2. Nodes $u$ and $v$ set $T^u_{d^\prime} = T^u_{d^\prime + 1} = t$ and $T^v_{d^\prime} = T^v_{d^\prime + 1} = t$. Note that both nodes $u$ and $v$ become singleton in level $d^\prime + 1$. Nodes $u$ and $v$ also set $T^u_i = T^v_i = max(T^u_i, T^v_i)$ for $i = d^\prime - 1, d^\prime - 2,\cdots,B_u$. For an example, check the timestamps of nodes $U$ and $V$ in $S_9$ (Figure \ref{fig:example}(c)).

\item[T2:] For a node $x$, let $c^\prime$ be the size of the longest common postfix between membership vectors $m(u^\prime)$ and $m(x)$ in $S_t$, where $u^\prime$ is the nearest communicating node (either $u$ or $v$) to $x$ in $S_t$. For example, in $S_8$ (Figure \ref{fig:example}(b)), $c^\prime$s for nodes $E$ and $G$ are 2 and 1, respectively. Let $M^x_d$ be the approximate median priority received by node $x$ at level $d$. For each level $d > \alpha$, each node $x$ with $G^x_d = G^u_d$ in $S_{t+1}$ (i.e. after group reassignment) sets $T^x_{d+1} = T^x_c$, where $c$ is the lowest level in $S_t$ with $\alpha \leq c < c^\prime$ such that $T^x_c > M^x_d$. If no such $T^x_c$ exists, node $x$ sets $T^x_{d+1} = M^x_d$. For an example, check the timestamps of node $E$ in $S_{t+1}$ (Figure \ref{fig:example}(c)) at levels 1 and 2, assuming $M^E_0 = 2$ and $M^E_1 = 5$.

\item[T3:] Let $x$ be a node such that $x \neq u, x \neq v, x \in l_\alpha$ with $G^x_\alpha = G^u_\alpha$ in $S_t$ (before transformation). Let $c^\prime$ be the size of the longest common postfix between $m(u)$ and $m(x)$ in $S_t$ and $c^{\prime \prime}$ be the longest common postfix between $m(u)$ and $m(x)$ in $S_{t+1}$. If $c^\prime - 1 > c^{\prime \prime} + 1$, each node $x$ sets $T^x_i = T^x_{c^\prime}$ for all $i =  c^\prime - 1,  c^\prime - 2, \cdots,  c^{\prime \prime} + 1$. Similarly, each $x$, $x \neq u, x \neq v, x \in l_\alpha$ with $G^x_\alpha = G^v_\alpha$ at time $t$ (before transformation), updates their timestamps w.r.t. node $v$. For example, for node $E$ in Figure \ref{fig:example}, $c^\prime = 3$ and $c^{\prime \prime} = 2$. Hence T3 does not apply for node $E$.

\item[T4:]  Each node $x$ that initialized or received $G_{lower}$ finds the lowest level $d$ such that $T^x_{d+1} = 0$. If such a $d$ exists and if $d > B_x$, node $x$ sets $T^x_i = T^x_{d+1}$ for $i = d, d-1, \cdots, B_y$.

\item[T5:] Let $x$ be a node, $x \in l_\alpha$ and $x$ belongs to a group $g$ at level $d$ in $S_t$, $d \geq \alpha$, such that $g$ splits into two subgroups at level $d$ in $S_{t+1}$. Each node $x$ sets $T^x_{d-1} = T^x_d$ only if $T^x_{d-1} = 0$.

\item[T6:] A \emph{group-base} of a node is the highest level at which the node belongs to its biggest group. For example, in the skip graph $S_8$ in Figure \ref{fig:example}(b), the group-base for node B is 1, as 1 is the highest level at which node $B$ is a member of its biggest group (B,G,D,U). Appendix \ref{sec:appen_group_id} presents details about how nodes maintain their group-base in a distributed manner. Let $B_x$ denote the group-base of a node $x$. Each node $x \in l_\alpha$ sets $T^x_d = 0$ for all $d < B_x$. For an example, check the timestamps of nodes $F$ and $I$ in $S_{t+1}$ (Figure \ref{fig:example}(c)) at level 1 and 0. Note that, $B_F = B_I = 2$ in $S_{t+1}$.

\end{enumerate}

After reassigning the timestamps, all nodes $x \in l_\alpha$ independently set themselves free for the next communication or transformation. A summary of the algorithm \textsf{DSG} is presented in Algorithm 1.


\begin{algorithm}[h] \label{alg:DSG}
\DontPrintSemicolon

\caption{Dynamic Skip Graph \textsf{DSG}\label{DSG} (summary)}

Upon request $(u,v)$ in $S_t$, establish the communication by using standard skip graph routing and find $\alpha$. Broadcast the membership vector, timestamps $T^u_d$, $T^v_d$, group-ids $G^u_d$, $G^v_d$, and group-bases $B_u$, $B_v$, where $d = \alpha, \alpha+1, ..., H_t$, to all nodes in $l_\alpha$.
\;
Each node $x \in l_\alpha$ computes their priority $P(x)$ using the using priority-rules P1, P2 and P3.
\;
Node $u$'s group at level $\alpha$ merges with node $v$'s group at level $\alpha$ by setting $G^x_\alpha = u$ for each group member $x$.
\;
Let $d = \alpha + 1$ and linked list $l_{d-1} = l_\alpha$. Compute the approximate median priority $M$ by using the algorithm \textsf{AMF}.
\;
Compute $|L_{low}|$, $|L_{high}|$, and $|g_s|$ using the balanced skip list formed by \textsf{AMF} if the condition in Equation \ref{eq:gr_split} is true.
\;
Determine the membership vector bit $V^x_d$ by using $P(x)$, $M$, $|L_{low}|$, $|L_{high}|$, and $|g_s|$ and update is-dominating-group $D^x_d$.
\;
Reuse the balanced skip list formed by \textsf{AMF} to check if a-balance property is being violated by the rearrangement. If yes, put a dummy node to break the chain violating the a-balance property.
\;
If a group at level $d$ is split into two subgraphs (because of step 6), find new (level $d$) group-id for the split group that moves to the 1-subgraph, and broadcast the new group-id by using the balanced skip list formed by \textsf{AMF}. Each $x$ in that group updates its $G^x_d$ with the new group-id. However, if the linked list formed by nodes that moved to 0-subgraph contains nodes $u$ and $v$, nodes $x$ set  $G^x_d = u$. Each node $x$ also updates their priorities ($P(x)$) using priority-rule P4 it its linked list does not contain nodes $u$ and $v$. Destroy the balanced skip list formed by \textsf{AMF}.
\;
Repeat steps 2 to 8 recursively and parallelly for all newly formed linked lists ($l_{d-1}$) that contains at least 2 nodes.
\;
Update group-ids and group-bases for involved nodes.
\;
Update timestamps using timestamps-rules T1-T6.
\;
Independently set nodes ($x \in l_\alpha$) free for next communication.
\;

\end{algorithm}


\presec
\subsection{Maintaining the a-balance property} \label{sec:a-bal}
\postsec
When a node $x$ moves to a new subgraph at level $d$, node $x$ checks if the a-balance property is being violated at level $d$ by the rearrangement of nodes. The balanced skip list created by \textsf{AMF} is reused to check if any consecutive $a$ nodes at level $d-1$ have moved to the same subgraph at level $d$. All nodes in the skip list at level $d$ check the new membership vector bit of $a$ nearby nodes at level $d-1$ in both sides and share this information with both neighbors of the skip list at level $d$ to detect chains. If a chain of size $a$ or longer is detected, a \emph{dummy node} is placed in the sibling subgraph at level $d$ to break the chain.

A dummy node is a logical node which requires $O(\log n)$ links and an identifier. A dummy node does not hold any data and only used for routing purpose. To implement dummy nodes, all regular nodes need to have the ability to handle extra $O(\log n)$ links.

When a dummy node is placed to break a chain, the identifier is picked by checking the identifier of a neighbor of the dummy node at level $d$ to ensure that all identifiers remain sorted at the base level of $S_{t+1}$. A dummy node does not participate in transformation and destroys itself when a transformation notification is received. While being destroyed, a dummy node simply links its left and right neighbors at all levels and deletes itself. The sole purpose of dummy nodes is to ensure that a-balance property is preserved after a transformation. Note that, the maximum number of dummy nodes possible is $n/a$.

\presec
\subsection{Node addition/removal} \label{sec:add-remove}
\postsec
Nodes can be added or removed by using standard node addition or removal procedures for skip graphs. When a node is added, the new node needs to initialize its variables with the default (initial) values. Following a node addition, the new node checks if the a-balance property is violated due to the node addition. Following a node deletion. a neighbor of the deleted node from each level checks if the a-balance property is violated due to node deletion. In case of a violation, a dummy node is placed to protect the a-balance property, as described in Section \ref{sec:a-bal}.

\presec
\section{Approximate Median Finding for Skip Graphs (AMF)}
\label{sec:amf}
\postsec

Given a linked list of size $n$ with each node holding a value, \textsf{AMF} is a distributed algorithm that finds an approximate median of the values in expected $O(\log n)$ rounds. Given that the size of the linked list is bigger than a constant $a$, we first construct a probabilistic skip list where the left-most node steps up to the next level with probability 1, and all other nodes step up to the next level with a probability $1/a$. After stepping up to the next level, nodes find their neighbors linearly from the level it stepped up. We write two consecutive nodes are supported by $k$ nodes if they have $k - 1$ nodes in between at the immediate lower level. When a linked list at some  level is built, nodes locally check if two consecutive nodes are supported by at least $a/2$ and at most $2a$ nodes. If two or more consecutive nodes are supported by less than $a/2$ nodes, they select the node with the highest identifier as a leader, and leader asks some nodes to step down to make sure each consecutive nodes in the list is supported by at least $a/2$ modes. Similarly, if two consecutive nodes are supported by more than $2a$ nodes, the node with the higher identifier asks some nodes in between to step up so that no two consecutive nodes in the list are supported by more than $2a$ nodes. The construction ends when the left-most node become the member of a singleton list (i.e. the root of the skip list) at some level. Let $l_i$ denote the linked list at level $i$ of the skip list. We refer the base level as level 0. Let $h$ be the height of the skip list, then level $h$ is the only level where the left-most node is singleton. The left-most node (i.e. root) broadcasts the value $h$ to all nodes of the skip list. Let $h = \log_b n$, then $2a \geq b \geq a/2$.

Median finding algorithm is a recursive algorithm that works in rounds. At the first round each node $x \in l_0$, $x \notin l_1$ forwards their values to their left neighbors, and any value received from the right neighbor is also forwarded to the left neighbor. This way values hold by all nodes that did not step up to level 1 are gathered to their nearest left neighbor that stepped up to level 1 (nodes in $l_1$). For implementation, while forwarding the value to the left neighbor, a node adds a ``last-node" tag with its value if it has an immediate right neighbor that has lifted to the upper level. When a node in linked list $l_1$ receives such tag, it can move to the next step knowing that it will not receive any more value from level 0.

Each node $x \in l_1$ is expected to have $a$ values including its own. All nodes $x \in l_1$ but $x \notin l_2$ forward all values they have to the nearest left neighbor that stepped up at level 2. Therefore, all nodes $x \in l_3$ are expected to have $a^2$ values. This process continues until nodes at level $\left \lceil \log _{a/2} h \right \rceil + 2$ (note that  $\log _{a/2} h = \log_{a/2}  \log _b n)$  receives all their values. Clearly each node $x \in  l_{\left \lceil \log _{a/2} h \right \rceil+ 2}$ must receive at least $(a/2)^{(\left \lceil \log _{a/2} h \right \rceil+ 2)} = \frac{a^2h}{4}$ values.

These values are sorted locally by the nodes $x \in  l_{\left \lceil \log _{a/2} h \right \rceil+ 2}$, and each node uniformly samples $ah$ values from the sorted list. Nodes keep only the sampled values for the next round and discard all other values they have. Nodes that did not step up to any further level forward their sampled values to the nearest left-neighbor that stepped up to the next level.  A similar tagging mechanism explained earlier can be used for the implementation purpose. It is important to note that this algorithm satisfies the $\mathcal{CONGEST}$ model since all messages are $O( \log n)$ in size. The  process of gathering values in the nearest left neighbor at upper level and sampling them continues recursively until the left-most node of the skip list receives all the (expected $a^2h$) values at level $h$ (the top level).

Each value forwarded by any node is attached with a \textit{left rank} and a \textit{right rank}. The left (right) rank attached with a value is the number of nodes in $l_\alpha$ that are guaranteed to have a larger (smaller) value than the value. Before every sampling, each node computes their \textit{left rank} and \textit{right rank}. Initially (at the base level) each node at the base level set both the left and right ranks attached with their value to zero. When a list of values are locally sorted by any node, the node computes the new left and right ranks for all the sampled values. The new left rank of a value is computed by adding the left ranks attached with all larger values in the sorted list. Similarly the new right rank of a value is computed by adding all the attached right ranks attached with the smaller values in the sorted list. Nodes forward their sampled values with the computed left and right ranks to the nearest left neighbor at the current level of the skip list.

When the left-most node at the top level of the skip list receives values from level $h - 1$ (the second highest level), it computes the median  based on the left and right rank attached with the values and then broadcasts the value to all nodes in $l_0$ as the approximate median.

The algorithm \textsf{AMF} is summarized in Algorithm 2.

\begin{algorithm}[h]
\label{alg:AMF}
\DontPrintSemicolon
\caption{Approximate Median Finding \textsf{AMF}\label{AMF}}

Construct a probabilistic skip list where the left-most node steps up to next level with probability 1 and all other nodes step up to next level with probability $1/a$, where $a$ is a constant and parameter for a-balance property. While the linked list at any level is being built, nodes locally ensure that no two consecutive nodes are supported by less than $a/2$ or more that $2a$ nodes.\;
All nodes $x \in l_0$ but $x \notin l_1$ forward their value and any value received from right neighbor (at base level of skip list) to the left neighbor.  A tagging mechanism explained in section \ref{sec:amf} can be used for the synchronization purpose. \;

\ForAll {levels $d = 1,2, ... , h-2$ in the probabilistic skip list, sequentially}{

	all nodes $x \in l_d$, $x \notin l_{d+1}$ forward (using level $d$ links) the values they have to the nearest left neighbor that stepped up to the level $d+1$.    \;

	\If{$d \geq \left \lceil \log _{a/2} h \right \rceil + 1$}{ all nodes in $l_{d+1}$ locally sort all the values they received, uniformly sample  $ah$ values from the sorted list, and compute new left and right ranks for all the sampled values. Nodes keep only the sampled values for the next level and discard all other values. }

}

All nodes at level $h-1$ (except the left-most node) forward their sampled $ah$ values to the left-most node. \;

The left-most node (and also the only node in level $h$) sort all the values it receives, computes new left and right ranks for all values, and finds the approximate median from the sorted list based on the left and right ranks. The approximate median is then broadcasted to all nodes of the base level.    \;

\end{algorithm}

\begin{lemma}
\label{ApproxMedianLemma}
Given a linked list of $n$ nodes (each with a value), the algorithm \textsf{AMF} outputs a value within the range of ranks $\frac{n}{2}  \pm \frac{n}{2a}$.
\end{lemma}

\begin{proof}
Let $m$ be the actual median value, and $m_l$ and $m_r$ be two consecutive values in the sorted list of values received by the left-most node of the skip list at the level $h$,  such that $m_l \geq m \geq m_r$. Clearly either $m_l$ or $m_r$ is picked as the approximate median, determined by their final right and left ranks. To prove this lemma, we shall quantify the maximum possible number of values discarded from range ($m_l$, $m_r$) at any level.

Let $S^x_d$ denote the set of values that node $x$ receives at level $d$ from nodes at level $d - 1$  (including own values of node $x$). From the construction of the skip list, it is ensured that no two consecutive nodes are supported by less than $a/2$ or more than $2a$ nodes. Therefore, size $|S^x_d|$ for any node $x$ at any level $d > \left \lceil \log _{a/2} h \right \rceil + 2$ must be in between $a^2 h/2$ and $2a^2h$. Let $I^x_d$ denote the sampling interval for node $x$ at level $d$. Obviously, $I^x_d =\frac{|S^x_d|}{ah} - 1$.

Each node $x$ from level $h-1$ (i.e. $x \in l_{h-1}$) contributes $ah$ values to the sorted list processed by the left-most node at level $h$. Therefore, any node $x \in l_{h-1}$ can discard at most $I^x_{h-1}$ values between $m_l$ and $m_r$. To maximize the number of values discarded between $m_l$ and $m_r$, each of $I^x_{h-1}$ values between $m_l$ and $m_r$ in $S^x_{h-1}$ must come from different nodes at the lower level $h-2$. Let $s \in S^x_{h-1}$, thus $s$ is one of the sampled values from node $y \in l_{h-2}$. if $m_l < s < m_r$, then there can be at most $2I^y_{h-2}$ values from range $(m_l,m_r)$  in $S^y_{h-2}$ (considering two sampling intervals from both sides). Hence, $S^y_{h-2}$ can have at most $2I^y_{h-2} + 1$ values at level $h-2$ from range $(m_l,m_r)$.

Again to maximize the number of discarded values from range $(m_l,m_r)$, all these $2I^y_{h-2} + 1$ values must come from all possible $I^y_{h-2} + 1$ nodes at level $h-3$; only one node from level $h-3$ contributes one (sampled) value, and each of rest of the $I^y_{h-2}$ nodes contributes 2 values each. Therefore, for any node $x \in l_{h-3}$, there can be at most $3I^x_{h-2} + 2$ values from range $(m_l,m_r)$  in $S^x_{h-2}$.

Similar reasoning can show that there can be at most $(h - k- 1) I^x_k + k$ values from range $(m_l,m_r)$  in $S^x_k$ for any node $x \in l_k$, where $k = \left \lceil \log _{a/2} h \right \rceil + 2$. However,  $S^x_k$ has $(ah + 1)I^x_k$ nodes and $k$ is the lowest level where values were discarded through sampling. Since $\frac{(h - k- 1) I^x_k + k}{(ah + 1)I^x_k} < \frac{1}{a}$, there are less than $n/a$ values in between $m_l$ and $m_r$. Thus at least one of the values between $m_l$ and $m_r$ must fall in the range of ranks $\frac{n}{2}  \pm \frac{n}{2a}$.
\end{proof}

\presec
\section{Analysis}
\label{sec:analysis}
\postsec

Appendix \ref{sec:appen_notation_table} presents a list of frequently used notations used in this section.

\begin{lemma}
\label{TimestampLemma}
Let $g_d$ be a group at level $d$ and $x$ be a node in $g_d$ in skip graph $S$. Let $G_x (V,E)$ be a communication graph where  $V$ is the set of all nodes in $S$ and $E$ represents only communications during the time period starting from time $T^x_d$, and ending at time $t$ (inclusive). All nodes $y \in g_d$ with $T^y_d > T^x_d$ are connected in $G_x$.
\end{lemma}

\begin{proof}
We present a proof by induction for this lemma. Let us assume that the lemma holds for all groups in $S_t$. We show that the lemma holds for $S_{t+1}$ as well. Let nodes $u_t$ and $v_t$ communicate at time $t$. By assumption, the lemma holds for all groups of nodes $u_t$ and $v_t$ in $S_t$. Now, since nodes $u_t$ and $v_t$ get connected in the communication graph at time $t$, priority rule P2 and timestamp rule T2 ensure that the lemma holds for any newly formed group containing nodes $u_t$ and $v_t$ in $S_{t+1}$.

Now, for any newly formed group in $S_{t+1}$ that does not contain nodes $u_t$ and $v_t$, there are two possibilities. The first possibility is that all nodes of such a new group had a positive priority during the transformation from $S_t$ to $S_{t+1}$ while receiving $M$ for level $d$. For this case, the lemma holds because (1) all these nodes were in at least one group of either $u_t$ or $v_t$ in $S_t$, and (2) timestamp rule T3 ensures that the new timestamps for the corresponding level are consistent with the lemma. The second possibility is that all node of the new group had a negative priority for level $d$ during the transformation. Such a new group can only be created by a split of a group of size greater than two-third of the number of nodes in the corresponding linked list. Because of the use of boolean variable is-dominating-group, these groups are identical of one of the groups in $S_t$ at level $d+1$. Hence the lemma holds.

Now we analyze the base case for induction. Clearly, the lemma holds for $S_1$ since all groups are the only member of their groups. Now, from the construction of the algorithm, it is easy to see that the lemma holds for $S_2$ as the groups with size $>1$ are the groups that contain nodes $u_1$ and $u_2$ and timestamp rules T1 ensures that the timestamps in $S_2$ are consistent.

\end{proof}

\begin{lemma}
\label{NegativeM}
Let $g_d$ be a group at level $d$ in $S_t$ such that for all pair of nodes $(x,y) \in g_d$, distance $d_{S_t} (x,y) = O(\log T_t (x,y))$, where $T_t (x,y)$ is the working set number for the node pair $(x,y)$ at time $t$. If $g_d$ is split into 2 new (sub)groups at level $d$ in $S_{t+1}$ due to a negative $M$, for all pair of nodes $(x,y) \in g_d$, $d_{S_{t+1}} (x,y) = O(\log T_{t+1} (x,y))$.
\end{lemma}
\vspace{-0.1in}
\begin{proof}
Let us consider a pair (x,y) such that the distance between nodes $x$ and $y$ increased due to the split (i.e. $d_{S_{t+1}} (x,y) > d_{S_t} (x,y)$). Then one of the nodes from pair $(x,y)$ must move to the 0-subgraph and the other node must move to the 1-subgraph at level $d$ in $S_{t+1}$. Let us assume that, node $x$ moves to the 0-subgraph and node $y$ moves to the 1-subgraph.

A group can split due to a negative $M$ only if the size of the group is bigger than two-third of the size of the corresponding linked list. From Section \ref{sec:transformation}, clearly nodes moving to the 1-subgroup (due to a split resulted by negative $M$) have $D^x_d = True$. Now, nodes set is-dominating-group as $True$ only on formation of a group (due to a positive $M$) that contains the communicating nodes. Hence, according to the timestamp rule T2, all nodes $x$ with $D^x_d = True$ have $T^x_d \geq M_p$, where $M_p$ is the positive priority used (in past) to set $D^x_d = True$. This implies that, after time $M_p$, the node $x$ did not communicate with any of the nodes in the group moving to 1-subgraph. According to the definition of working set number, this implies that at least $T_t (x,y)$ nodes move to the 1-subgraph due to the split resulted by a negative $M$.

We construct a communication graph $G_M$ with communications during the time period starting from time $M_p$ and ending at the current time $t$ (inclusive). According to Lemma \ref{TimestampLemma}, all nodes moved to the 1-subgraph are connected in $G_M$. Lemma \ref{ApproxMedianLemma} implies that the number of nodes moved to the 1-subgraph is at most $\frac{n^\prime}{2}  + \frac{n^\prime}{2a}$, where $n^\prime$ is the size of the corresponding linked list. This follows:
\begin{equation} \label{eq:nmEq2}
d_{S_{t+1}} (x,y) \leq a \log_{3/2} \bigg( \frac{n^\prime}{2}  + \frac{n^\prime}{2a} \bigg) + a = O(\log T_{t+1} (x,y))
\end{equation}

\end{proof}

\begin{lemma}
\label{UVHeightLemma}
Given that nodes $u$ and $v$ communicate at time $t$, there exists a direct link between nodes $u$ and $v$ in $S_{t+1}$ at a level no higher than $\log _{\frac{2a}{a+1}} n$.
\end{lemma}
\vspace{-0.1in}
\begin{proof}
As described in the case 1 in Section \ref{sec:transformation}, any group containing communicating nodes $u$ and $v$ at any level $d$ can split at level $d + 1$ only if the nodes in the group receive a positive $M$. Now, according to Lemma \ref{ApproxMedianLemma}, a positive $M$ can split a subgraph of size $n$ into two subgraphs where size of any of the new subgraphs is at most $\frac{n}{2} + \frac{n}{2a}$.  Since communicating nodes always receive a positive $M$, the maximum possible height in $S_{t+1}$ at which nodes $u$ and $v$ move to a subgraph of size 2 is $\log _{\frac{n}{\frac{n}{2} + \frac{n}{2a}}} n = \log _{\frac{2a}{a+1}} n$.

\end{proof}

\begin{lemma}
\label{HeightLemma}
The maximum possible height after a transformation by \textsf{DSG} is $\log _{\frac{3}{2}} n$.
\end{lemma}
\vspace{-0.1in}
\begin{proof}
Algorithm \textsf{DSG} split a subgraphs into two smaller subgraphs at the immediate upper level. From Section \ref{sec:transformation}, it is easy to see that a subgraph of size $n$ can split into two subgraphs where the size of any of the new subgraphs is at most $\frac{2n}{3}$. Therefore, the maximum possible height of the skip graph after a transformation is $\log _{\frac{3}{2}} n$.

\end{proof}

\vspace{-0.1in}
\begin{theorem}
\label{WSTheorem}
At any time $t$, given that any two nodes $u$ and $v$ communicated earlier, the distance between $u$ and $v$ in skip graph $S_{t}$ is $O(\log T_t (u,v))$, where $T_t (u,v)$ is the working set number for the node pair $(u,v)$ at time $t$.
\end{theorem}
\vspace{-0.1in}
\begin{proof}

Let $t^\prime$ be a time between $t$ and the last time $u$ and $v$ communicated. Let $k$ be the longest common postfix between $m(u)$  and $m(v)$ in skip graph $S_{t^\prime}$, then the distance between $u$ and $v$ in $S_{t^\prime}$ is at most $ak$. Suppose node $u_t$ communicates with node $v_t$ at time $t^\prime$. Let $z = H_t - k$, then there exists a linked list $l_z$ at level $z$ such that $u \in l_z$, $v \in l_z$. Now one of the following four cases must occur:


Case 1: $u_t \in l_z$, $v_t \in l_z$, and $u_t$ and $v_t$ are in the same  linked lists at level $z+1$. Clearly, the distance between $u$ and $v$ remains unchanged in skip graph $S_{t+1}$.

Case 2: $u_t \in l_z$, $v_t \in l_z$, and $u_t$ and $v_t$ are in two different linked lists at level $z+1$. Clearly, the distance between $u$ and $v$ cannot be increased in skip graph $S_{t+1}$.

Case 3: $u_t \in l_z$, $v_t \notin l_z$ or vice versa. We analyze only the case when $u_t \in l_z$, $v_t \notin l_z$, and analysis for the opposite case $u_t \notin l_z$, $v_t \in l_z$ is similar. Since $u$ and $v$ communicated earlier, they must hold the same group-id for level $z$, i.e. $G^u_z = G^v_z$. Let $g^u_z$ be the group at level $z$ that contains both nodes $u$ and $v$. Let $\alpha$ be the highest common level in $S_{t^\prime}$ for communicating nodes $u_t$ and $v_t$. Then if $u_t \not \in g^u_z$, the rearrangement will not increase the distance between $u$ and $v$ unless the group $g^u_z$ splits at level $z$ due to a negative $M$ (case 2 in Section \ref{sec:transformation}). According to Lemma \ref{NegativeM}, the distance between $u$ and $v$ in skip graph $S_{t^\prime + 1}$ remains $O(\log T_{t^\prime + 1} (u,v))$, even if group $g^u_z$ splits at level $z$ due to a negative $M$.

Now, if $u_t \in g^u_z$, let $X = \{x \in l_\alpha | P(x) > \min (P(u),P(v))\}$. we argue that $|X| \leq T_{t^\prime + 1} (u,v)$, which proves the lemma for this case since $d_{S_{t^\prime+1}} (u,v) = O(\log |X|)$. Let $l^u_{\alpha+1}$ and $l^v_{\alpha+1}$ be the linked lists in $S_{t^\prime + 1}$ at level $\alpha + 1$ such that $u_t \in l^u_{\alpha+1}$ and $v_t \in l^v_{\alpha+1}$. We construct a communication graph $G^\prime$ for communications during the time period between the time $\min (P(u),P(v))$ and the current time $t$ (inclusive). According to Lemma \ref{TimestampLemma}, all nodes $x \in l^u_{\alpha+1}$ with $P(x) > \min (P(u),P(v))$ are connected in the communication graph $G^\prime$. Similarly, all nodes $x \in l^v_{\alpha+1}$ with $P(x) > \min (P(u),P(v))$ are connected in communication graph $G^\prime$ as well. Now since nodes $u_t$ and $v_t$ communicate at time $t^\prime$, all nodes $x \in X$ must be connected in $G^\prime$. Based on our definition of working set number, $|X| \leq T_{t^\prime + 1} (u,v)$ follows.

Case 4: $u_t \notin l_z$, $v_t \notin l_z$. Two possible scenarios under this case: (1) neither node $u_t$ nor node $v_t$ belongs to the group of nodes $u$ and $v$ at any level; and (2) node $u_t$, or node $v_t$, or both nodes $u_t$ and $v_t$ belong to the group of nodes $u$ and $v$ at a level lower than $z$. The first scenario is equivalent to the scenario $G^{u_t}_z \neq G^u_z$ described in case 3. To analyze the second scenario, let $z^\prime$ be the level such that $G^{u_t}_{z^\prime} = G^u_{z^\prime} = G^u_{z^\prime}$. Now as transformation takes place recursively at different levels, the scenario is equivalent to the scenario $u_t \in g^u_z$ described in case 3 as long as both nodes $u$ and $v$ move to the 0-subgraph. However, if both nodes $u$ and $v$ move to the 1-subgraph at some level, the scenario becomes equivalent to the scenario $G^{u_t}_z \neq G^u_z$ described in case 3.

\end{proof}

\vspace{-0.1in}
\begin{theorem}
\label{WSTheorem2}
Given a communication sequence $\sigma$, the expected running time of algorithm \textsf{DSG} is $(WS(\sigma))^2$ rounds.
\end{theorem}
\vspace{-0.1in}
The proof relies on the fact that the most expensive operation performed by \textsf{DSG} is to find the approximate median priority for all involved levels for communication request $\sigma_t$. The most expensive operation performed by a single call of algorithm \textsf{AMF} is to construct the balanced skip list. Due to the randomization involved in construction, the expected time to construct a balanced skip list is $O(h)$, where $h$ is the height of the skip list.


\vspace{-0.05in}
\begin{theorem}
\label{WSTheorem4}
The routing cost for algorithm \textsf{DSG} is at most a constant factor more than the the amortized routing cost of the optimal algorithm.
\end{theorem}
\vspace{-0.1in}
\begin{proof}
Follows from Theorems \ref{WSTheorem3} and \ref{WSTheorem}.
\end{proof}

\vspace{-0.1in}
\begin{theorem}
\label{WSTheorem5}
The cost for algorithm \textsf{DSG} is at most logarithmic factor more than the amortized cost of the optimal algorithm.
\end{theorem}
\vspace{-0.1in}
\begin{proof}
Follows from Theorems \ref{WSTheorem3} and \ref{WSTheorem2}.
\end{proof}

\presec
\section{Conclusion} \label{sec:conclusion}
\postsec

We present a self-adjusting algorithm for skip graphs that relies on the idea of grouping frequently communicating nodes at different levels and using timestamps to determine nodes' attachments with their groups. We believe this study will lead to a general framework for distributed data structures with overlapping tree-like structures. 

Our algorithm \textsf{DSG} can be useful in networks where multiples levels are involved. For example, VM migration problem in data centers with levels such as rack-level, intra- and inter-data-center level, inter-continental level etc. Moreover, while the amortized routing time for most self-adjusting data structures is $O(\log n)$, our algorithm \textsf{DSG} guarantees $O(\log n)$ routing time for each of the individual communication requests. Thus, compared to most other self-adjusting networks, \textsf{DSG} is better suited for the cases where there is a time limit associated with each of the communications.
%


\bibliographystyle{abbrv}
\bibliography{sigproc}
\vfill\eject
\appendix
\section{Appendix} \label{sec:appen}
\presec
\subsection{Frequently Used Notations} \label{sec:appen_notation_table}
\postsec
\begin{table}[htbp]
\scriptsize
\begin{tabular}{|c|l|}
\hline
\textbf{Notation} & \textbf{Description} \\[1ex]
\hline
$m(x)$ & membership vector of node x \\
$V^i_j$ & \textit{membership vector bit} of node i for level j \\
$T^i_j$ & \textit{timestamp} for node i at level j \\
$G^i_j$ & \textit{group-id} for node i at level j\\
$D^i_j$ & \textit{is-dominating-group} for node i at level j\\
$B_x$ & \textit{group-base} of node x \\
$l_a$ & a linked list at level a  \\
$S_t$ & skip graph at time t\\
$M$ & approximate median priority\\
$T_i(x,y)$ & The working set number for node pair $(x,y)$ at time $i$\\
$WS(\sigma)$ & $\sum_{i=1}^{m}$ log$( T_i(\sigma_i))$, where $\sigma = \sigma_1, \sigma_2, \dots \sigma_m$\\
$\sigma_i$ & Communication request $(u_i,v_i)$ at time $i$ \\
$d_{S_i} (x,y)$ & Distance between nodes $x$ and $y$ in skip graph $S_i$ \\

\hline
\end{tabular}
\caption{Notations}
\label{tab:notation}
\end{table}

\presec
\subsection{Standard Skip Graph Routing} \label{sec:appen_sg_routing}
\postsec

The standard skip graph routing \cite{SG} works as follows. Routing starts at the top level from the source node and traverses through the skip graph structure. If the identifier of the destination node is greater than that of the source node, then at each level, routing moves to the next right node until the identifier of the next node is greater then the identifier of the destination node. When a node with an identifier greater than the destination node is found, the routing drops to the next lower level, continuing until the destination node is found. If the identifier of the destination nodes is smaller than that of the source node, routing takes place in the similar manner expect it moves to the next left node instead of right, and drops to the lower level when a node with smaller (instead of greater) identifier is found.


\presec
\subsection{Updating group-ids for levels below $\alpha$} \label{sec:appen_group_id}
\postsec

\textsf{DSG} requires each node to store a number (between 0 and $H_t$), referred to as \emph{group-base}. The group-base for a node is the highest level at which the node belongs to its biggest group. For example, in the skip graph $S_8$ in Figure \ref{fig:example}(b), the group-bases for nodes H,F,B and G are 3,2,1 and 1, respectively. We use the notation $B_x$ to denote the group-base of any node $x$. Initially, before any communication, group-base for each node is set to the lowest level at which the node is singleton. Also, the notification message sent after routing includes the group-bases for both nodes $u$ and $v$, along with group-ids, timestamps, and membership vectors.

Each node $x \in l_\alpha$ with $G^x_\alpha = u$  initializes a vector $G_{lower}$ as follows:
\[
    G_{lower}=
\begin{cases}
    [G^u_{0}, G^u_{1}, \dots,  G^u_{\alpha-1}],& \text{if } B_u \leq B_v\\
    [G^v_{0}, G^v_{1}, \dots,  G^v_{\alpha-1}],& \text{otherwise}
\end{cases}
\]

Node $u$ broadcasts a message  $\big \langle G_{lower}, \min (B_u, B_v), G^u_{\max (B_u, B_v)}, G^v_{\max (B_u, B_v)} \big \rangle$ to all nodes $y \in l_{\max (B_u, B_v)}$ such that $u,v \in l_{\max (B_u, B_v)}$. Each such node $y$ with $G^y_{\max (B_u, B_v)} = G^u_{\max (B_u, B_v)}$ or $G^y_{\max (B_u, B_v)} = G^v_{\max (B_u, B_v)}$ updates their group-base by setting $B_y = \min (B_u, B_v)$ and updates group-ids $G^y_i = G_{lower}[i]$ for $ i = 0, 1, \dots \alpha - 1$.

Regardless of the outcome of the comparison $\min (B_u, B_v) < \alpha$, each node $x \in l_\alpha$ with $G^x_\alpha = u$ sets group-ids $G^x_i = G_{lower}[i]$ for $ i = 0, 1, \dots \alpha - 1$. Moreover, each node $x \in l_\alpha$ updates its group-base $B_x$ as follows:
\begin{itemize}

\item[--] If $x$'s group at any level $d$ ($d \geq \alpha$) splits into 2 subgroups due to transformation, and if $B_x = d$, $x$ sets $B_x = B_x - 1$.

\item[--] Let $d$ be the lowest level at which $x$'s group splits due to transformation. if $B_x = \alpha$ and $d > \alpha+1$, $x$ sets $B_x = d - 1$.

\end{itemize}

It is important to understand that if $G^u_{\alpha-1} \neq G^v_{\alpha-1}$, then the working set number for the node pair $(u,v)$ is greater than the routing distance for $(u,v)$.

\presec
\subsection{Distributed Sum Using a Skip List} \label{sec:appen_dist_count}
\postsec
Each node holds a number and we want to compute the sum of the numbers held by all the nodes. Each node of the base level of the skip list forwards their number to the nearest neighbor that steps up to the upper level of the skip list. Any node receiving numbers from the neighbors from lower level computes the sum of the numbers and forwards the sum to the nearest neighbor stepping up to the upper level. As this happens recursively at each level, the head node of the skip list computes the final sum in $O(\log n)$ rounds and then broadcasts the sum to all the nodes.

\end{document}